\documentclass[12pt]{article}
\newcommand{\qed}{\rule[-0.2ex]{0.3em}{1.4ex}}

{\par\vspace{0.5ex}\noindent{\bf Proof:}\hspace{0.5em}}%
{\nopagebreak
\strut\nopagebreak
\hspace{\fill}\qed\par\medskip\noindent}
\newenvironment{proof}% lowercase version...
{\par\vspace{0.5ex}\noindent{\bf Proof:}\hspace{0.5em}}%
{\nopagebreak
\strut\nopagebreak
\hspace{\fill}\qed\par\medskip\noindent}
{\par\vspace{0.5ex}\noindent{\bf Proof Attempt:}\hspace{0.5em}}%
{\nopagebreak
\strut\nopagebreak
\hspace{\fill}\qed\par\medskip\noindent}
\newtheorem{theorem}{Theorem}

\newtheorem{lemma}{Lemma}

\newtheorem{definition}{Definition}
\newtheorem{corollary}{Corollary}

\usepackage{amsmath}
\usepackage{amsfonts}
\usepackage{graphicx}

\usepackage[numbers,sort&compress]{natbib}
\DeclareMathOperator*{\argmin}{arg\,min}

\newcommand{\IR}{{\rm\hbox{I\kern-.15em R}}}
\newcommand{\reals}{{\rm\hbox{I\kern-.15em R}}}
\newcommand{\IN}{{\rm\hbox{I\kern-.15em N}}}
\newcommand{\IZ}{{\sf\hbox{Z\kern-.40em Z}}}

\newcommand{\E}{\mbox{\sf {I\kern-.15em E}}}

\newcommand{\Prob}{\mbox{\sf \hbox{I\kern-.15em P}}}

\begin{document}
\bibliographystyle{ieeetr}

\title{{\bf Foundations for Wash Sales}}
\author{
Phillip G. Bradford\thanks{University of Connecticut at Stamford, Department of Computer Science and Engineering, Stamford Connecticut. phillip.bradford@uconn.edu, phillip.g.bradford@gmail.com} 
}

\maketitle

\begin{abstract}
Consider an ephemeral sale-and-repurchase of a security resulting in the same position before the sale and after the repurchase.
A sale-and-repurchase is a wash sale if these transactions result in a loss within $\pm 30$ calendar days.
Since a portfolio is essentially the same after a wash sale, any tax advantage from such a loss is not allowed.
That is, after a wash sale a portfolio is unchanged so any loss captured by the wash sale is 
deemed to be solely for tax advantage and not investment purposes.

This paper starts by exploring variations of the birthday problem to model wash sales. 
The birthday problem is: Determine the number of independent and identically distributed random variables required so there is a probability of at
least $\frac{1}{2}$ that two or more of these random variables share the same outcome.
This paper gives necessary conditions for wash sales based on variations on the birthday problem.
Suitable variations of the birthday problem are new to this paper.
This allows us to answer questions such as:
What is the likelihood of a wash sale in an unmanaged portfolio where purchases and sales are independent, uniform, and random?
Portfolios containing options may lead to wash sales resembling these characteristics.
This paper ends by exploring the Littlewood-Offord problem as it relates capital gains and losses with wash sales.
\end{abstract}
\section{Introduction}
\label{Intro}

Wash sales impact a portfolio's tax liabilities.
Determining the likelihood of wash sales is also important for understanding investment strategies
and for comparing actively and passively managed portfolios.
Wash sales apply to investors, but not to market makers. 

Taxes play a significant role in economics and finance.
Taxes influence behavior, shape the engineering of financial transactions,
and sometimes have unintended consequences.
Therefore, thoughtful analysis is imperative for taxes.
This paper adds firm mathematical foundations to aid the understanding of wash sale taxes.

The main goal of this paper is:
To provide foundations for certain wash sales - in cases when they may occur as well as the capital gain implications.
This may also help differentiate managed funds and unmanaged index funds in terms of wash sales.

Wash sales are sometimes created by the exercise of options, thus a portfolio manager may not be able to avoid a wash sale in some contexts.
For example, suppose an in-the-money American-style put option is written in a portfolio.
Provided this option remains in-the-money, it may be exercised by its holder\footnote{Options, like shares of stock, are fungible
and there are specific option exercise assignment allocation methods used to allocate exercised options~\cite{FINRA}.}
at anytime up to its expiry.
If the exercise of this put option replaces shares sold at a loss in the prior 30 days, then
this is a wash sale.
This option's exercise is beyond control of the portfolio manager. 

The foundations given here start with variations of the classical birthday problem from probability theory~\cite{vM,F,CLRS}.
This work has implications on wash sales.
Also, the Littlewood-Offord problem~\cite{LO,E1945,TV} is applied to understand capital gains for certain wash sales.
The Littlewood-Offord problem is viewed from the perspective of the probabilistic method.

For convenience, let $[n] \equiv \{ \ 1, 2, \cdots, n \ \}$.

Suppose a security is sold at a loss on day $d_2$. This sale is a wash sale if substantially the same security is purchased
within $\pm 30$~calendar days from $d_2$, see for example~\cite{IRS550}.

\begin{definition}[US wash sale~\cite{IRS550}]
{\sf
Consider three dates $d_1, d_2, d_3 : d_1 \leq d_2$ and $d_1 \leq d_3$ where $|d_3 - d_2| \leq 30$ calendar days.

Suppose $s$ shares of a security are purchased on date $d_1$ at price $p_1$.

At some later date $d_2$, $s$ shares are sold for price $p_2 < p_1$. Thus, the $s$ shares are sold at a loss.
Then within $\pm 30$ days on date $d_3$, $s$ shares are repurchased for price $p_3$.

Since $|d_3 - d_2| \leq 30$ days, then the next adjustments must be made~\cite{IRS550}:

\begin{enumerate}

\item The loss $p_1 - p_2$ is not permissible for taxes. That is, this loss may not be subtracted from profits or gains and it may not be used to get a 
lower tax rate.

\item The cost-basis of the shares repurchased on $d_3$ is set to $p_3 + (p_1 - p_2)$. The shares purchased on $d_3$ have the start of their holding 
period reset to $d_1$.

\end{enumerate}
}
\label{wash-sale-definition}
\end{definition}

%
%If $s$ shares are sold at a loss, within 30 days there may be a number of purchases adding $s$ shares back to the portfolio.
%This is a single wash sale of $s$ shares.
%
%The analysis in this paper handles this case as well as the simpler one-sale and one-purchase case.
%

Short positions may also be wash sales. For example, consider holding a short position of 100 shares of a security starting on date $d_1$ in a 
portfolio $\Pi$.
Then suppose this short position is closed at a loss by purchasing 100 shares on day $d_2$. Once this position is closed on day $d_2$,
then $\Pi$ contains no shares of this security.
Next re-short another 100 shares of substantially the same security on $d_3$ where $|d_3 - d_2| \leq 30$ days.
These transactions leave the portfolio the same while getting a tax advantage for the loss. This tax advantage is also disallowed by
the wash sale rules.

%
%Wash sale tax rules disallow taking losses that only give tax advantage, but do not change the underlying portfolio.
%
%The holding-period start reset for a wash sale to $d_1$ in addition to the change of cost-basis to $p_3 + (p_1 - p_2)$ is to reverse
%any potential gain from a wash sale. 
%

Consider a wash sale as described by Definition~\ref{wash-sale-definition}, where $(p_1 - p_2) + p_3 > p_1$ or in other words $p_3 > p_2$. 
Suppose the shares are sold at price $p_4 > (p_1 - p_2) + p_3 > p_1$ at
the later date $d_4 \geq d_3$. 
In the case with the wash sale, there is a capital gain of $p_4 - \left[ (p_1 - p_2) + p_3 \right]$ which is smaller than
the capital gain $p_4 - p_1$ if the wash sale had not occurred.
Capital gains are taxable.
A capital gain $p_4 - p_1$ is from the single purchase of the shares for price $p_1$ on $d_1$ and the single sale of the shares
on date $d_4$ for price $p_4$, thus skipping the sale at a loss and repurchase.

This means such a wash sale gives $p_4 - p_1 - \left[ p_4 - \left[ (p_1 - p_2) + p_3 \right] \right]$ or $p_3 - p_2$ less taxable income
than a single purchase of the security at price $p_1$ on date $d_1$ and a single sale for price $p_4$ on $d_4$.
Of course, a wash sale's loss is not allowed.

%
%The reset of the holding-period start to $d_1$ may also give advantage to special cases for wash sales.
%Just as an investor may only purchase a security on $d_3$ at price $p_3$ and only sell it on $d_4$ for $p_4$, but if $d_4 - d_3$ gives a higher short-term capital 
%gains tax-rate than the long-term capital gains tax-rate from $d_4 - d_1$, then $p_4 - p_3 > p_4 - \left[(p_1 - p_2) + p_3 \right]$ or $p_2 - p_1$.
%Then the wash sale...
%
%
%Moreover, the case where $p_3 \gg p_2$ can occur, for example, on a take-over, a stock's move to a more prestigious index, etc.
%

Wash sales may be avoided by restricting each security in a portfolio to be either purchased or sold only every 31 calendar days.
This restriction may not be suitable for many portfolios.
In a portfolio containing options, it may be impossible to maintain this restriction.

It has also been suggested, e.g.~\cite{JM}, wash sales may be avoided by purchasing or selling (moderately) correlated, but not substantially the same,
securities.
That is, if a security is sold at a loss then purchase a different but correlated security within 30 days maintaining some of a portfolio's characteristics
while keeping the tax advantage.

Historically many securities are assumed to only trade on about $n=252$ business days per year~\cite{H}.
Although reflecting on global markets one may assume there are
$n=365$ trading days.

\subsection{Background}

There has not been much research on wash sales, e.g.,~\cite{JM}.
There is important work on taxation and its investment implications.
Take, for example, ~\cite{C83,C84} and~\cite{DSZ}.

The birthday problem is classical. 
According to a blog post by Pat B~\cite{PWorld}
the birthday problem may have originally been given by Harold Davenport as cited in~\cite{BC} and 
later published by~\cite{vM}.
In any case, von Mises gave the first published version to the best of our knowledge.

Bounds of day counts for the birthday problems include~\cite{N} who gives bounds for birthdays of 
distance $d$ for both linear years as well as cyclic years. 
In a cyclic year, 1-January is a single day from 31-December of the same year.
Bounds for birthdays of distance $d$ for cyclic years are given by~\cite{AM}.

The birthday problem applied to boys and girls (random variables with different labels) are discussed in~\cite{CN} 
as well as~\cite{P}.
That is, how many birthdays are shared by one or more boys and one or more girls?
A comprehensive view is provided by~\cite{NS} including stopping problems with the boy-girl birthday problem.
Non-uniform bounds for online boy-girl birthday problems are given by~\cite{GH} and~\cite{S}.

Tight bounded Poisson approximations for birthday problems are given by~\cite{CDM}.
Poisson approximations to the binomial distribution for the boy-girl birthday problem is given by~\cite{P}.
 A Stein-Chen Poisson approximation is used by~\cite{AGG} to solve variations of the standard birthday problem.
Matching and birthday problems are given by~\cite{DG}.
Incidence variables are used to study birthday problems with Pareto-type distributions in~\cite{BPSW}.

Applications of the birthday problem include: computer security~\cite{NS,Stinson,BPSW,GH},
public health and epidemiology~\cite{SGMG}, psychology, DNA sequence alignment, experiments, and games~\cite{DM,DG_m}.
Summaries of work on the birthday problem are in~\cite{DG_m}, ~\cite{DM}, and~\cite{R}.

Results on the expectation for getting $j$ different letter $k$-collisions are given by~\cite{FGT}.
Their results are expressed as truncated exponentials or gamma functions.

\subsection{Structure of this Paper}

Section~\ref{bday_problem} reviews variants the birthday problem applied here. 
First the classical birthday problem is discussed. 
Next this section progresses through the $\pm d$ birthday problem.
After the definition and key results are given about the $\pm d$ birthday problem, the
boy-girl birthday problem is explored.
Finally, the $\pm d$ boy-girl birthday problem is defined and several bounds are derived as they relate to a necessary condition for
wash sales.

Subsection~\ref{Wash-sale-example-1} gives an example of wash sales based on boy-girl birthday collisions of a single day.

Section~\ref{General-Wash-Sales} generalizes results of the previous sections. In particular, it shows how to compute
$B_d(n,b,g)$, the number of $b$ boys and $g$ girls that give a probability of $\frac{1}{2}$ or more where a boy and a girl have birthdays
within $d$ days of each other over $n$ days.

Subsection~\ref{Wash-sale-example-2} gives an example of wash sales based on boy-girl birthday collisions over a range of $\pm d = 30$ days.

Finally, Section~\ref{VarCGCL} explores how wash sales impact capital gains and losses.
Since wash sales are capital losses, they may offset capital gains.
Several results, including the Littlewood-Offord problem, are applied to capital gains and losses as they may be impacted by wash sales.

\section{The Birthday Problem and Wash Sales}
\label{bday_problem}

The birthday problem is often applied to finding the probability of coincidences. So there is a rich literature on variations of the 
birthday problem~\cite{DM,DG_m}.
Asset sales are often viewed as carefully selected.
However, portfolios using American-style options may exhibit asset sales or purchases beyond the control of the portfolio managers.

\begin{definition}[Birthday-Collision]
{\sf
Given two random variables $X_1, X_2$ mapping respectively to $x_1, x_2$ in the same range $[n]$, then
a {\em birthday-collision} is when $x_1 = x_2$.
}
\end{definition}

To model random wash sales, this paper assumes independent identically distributed random variables. 
A common statement of the birthday problem is:

\begin{definition}[Birthday Problem]
{\sf
Consider $n$ days in a year and $k$ independent identically distributed ({\em iid}) uniform random variables whose range is $[n]$ and $n \geq k$. 
What is the probability $B(n,k)$ of at least one birthday-collision among these $k$ random variables?
}
\label{standardBirthday}
\end{definition}

A key question is: Over $n$ consecutive days for what integer $k$ does
$\displaystyle \argmin_{k} \left\{ B(n,k) \geq \frac{1}{2} \right\}$ hold for $k$ iid uniform random variables?
In other words, given $n$ days, what is the least $k$ iid uniform random variables so that $B(n,k) = \frac{1}{2}$ ?

Solutions to this basic variation of the birthday problem are well known. 
The probability $B(n,k)$ is the compliment of the probability of $k$ iid uniform random variables having no birthday-collisions.
Therefore, if there are no birthday-collisions, then $k$ birthdays can be in ${n \choose k} \, k$! permutations
out of all possible $n^k$ mappings of the $k$ random variables onto $[n]$.
In other words, the ${n \choose k}$ subsets of $k$ distinct elements of $[n]$ is the exact number of subsets the $k$ variables 
may map to without a collision. These $k$ variables may be ordered in $k!$ permutations.
That is,

\begin{eqnarray*}
B(n,k) & = & 1-{n \choose k} \frac{k!}{n^k} \ \ = \ \ 1-\frac{n!}{(n-k)!}\cdot\frac{1}{n^k},
\end{eqnarray*}

\noindent
for $n \geq k$ and $B(n,k) = 1$ otherwise.

Starting with $n$ and a probability $p = B(n,k)$, then computing $k$ is often done using the inequality
$1 - x \leq e^{-x}$.
In particular, the smallest $k$ giving a probability of $\frac{1}{2}$ that there is at least one birthday-collision requires $k$ to be
roughly $\sqrt{2(\ln 2)n}$ or about $1.18\sqrt{n}$.
See for example,~\cite{vM,MR95,MU}.

Another classical approach is to look at the random variable $X$ as the sum of all birthday-collisions of $k$ people over $n$ days,
see for example~\cite{P85,BHS,P,DG}.
A concise exposition is given in~\cite{P85} which we follow.
Presume the birthday day of person $i \in [k]$ is given by the random variable $Y_i \in [n]$.
Since a potential birthday collision is a Bernoulli trial, so $X$ is binomially distributed.
Thus, $X \in \{ \, 0, 1, 2, \cdots, {k \choose 2} \, \}$ where ${k \choose 2}$ is the maximum number of potential birthday-collisions.
The expectation of the maximum number of birthday collisions possible is ${k \choose 2}$ with probability 
$\frac{1}{n} = \Prob[ Y_i =t  | Y_j = t], \ \ t \in [n]$ where $\{ i,j \} \subseteq [k]$.
The expected maximum number of birthday-collisions is $\frac{1}{n}{k \choose 2}$.
If $n$ is sufficiently larger than $k$, then $X$ is approximately Poisson where $\lambda = \frac{1}{n}{k \choose 2}$.
Thus, $\Prob[X \geq 1] \approx 1 - e^{-{k \choose 2}/n}$.

In the case of the $\pm d$ birthday problem, if two random variables $X_1, X_2$ map within $d$ days of each other, then
this is a $\pm d$ birthday-collision~\cite{N}.

Two birthdays $x_1$ and $x_2$ of distance $|x_1 - x_2|$ demark a span of size $1 + |x_1 - x_2|$.
For example, $|4\_{\mbox{\bf July}} - 3\_{\mbox{\bf July}}| = 1$, so these dates are in a $\pm d = \pm 2$ span, but not in a span of $\pm d = \pm 1$.

The next definition is based on~\cite{N,AM,DM}.

\begin{definition}[$\pm d$ Birthday Collisions] % P...?
{\sf
Consider $n$ days in a year, spans of less than $\pm d$ days, and $k$ iid uniform random variables with range $[n]$: Then $B_d(n,k)$ is the probability
at least two such random variables have a $\pm d$ {\em birthday-collision}. 
That is, these two random variables have ranges in {\em less than} $d$ days of each other.
}
\end{definition}

In $n$ days with a $\pm d$ span, then $\displaystyle \argmin_k \left\{ B_d(n,k) \geq \frac{1}{2} \right\}$ gives the smallest $k$
so there is a probability of at least $\frac{1}{2}$ where at least two such random variables are fewer than $d$ days from each other.

\begin{definition}[Blocks of days] % 
{\sf
Let $i: k > i > 1$.
Suppose birthdays are ordered as $x_1 \leq x_2 \leq \cdots \leq x_k$, then for a birthday $x_i$ its {\em nearest birthday pairs}
are $(x_{i-1}, x_i)$ and $(x_i, x_{i+1})$.
There are no birthdays between $x_{i-1}$ and $x_i$ and
there are no birthdays between $x_i$ and $x_{i+1}$.

A {\em block} of days contains a single birthday on one of its end-points. 
The birthday $x_i$ is associated with two blocks: $(x_{i-1}, x_i]$ and $[x_i, x_{i+1})$.
}
\end{definition}

The days between $x_1$ and $x_2$ form a block of size $|x_1 - x_2|$ since there are no birthdays between $x_1$ and $x_2$.
Thus, two nearest birthday pairs contained in a span of $\pm d$ are separated by a block of size $d-1$.

Take $k$ iid uniform random variables and consider $\pm d$ birthday-collisions over $[n]$ days.
Naus~\cite{N} gives the next idea: If there are no $\pm d$ birthday-collisions, then there must be at least size $d-1$ blocks of no birthdays between
each nearest birthday pair. This gives a total of $(k-1)(d-1)$ days with no birthdays in $k-1$ contiguous blocks of at least $d-1$ days each.
Therefore, if there are no $\pm d$ birthday-collisions, then $k$ birthdays can be in ${n - (k-1)(d -1) \choose k} \, k!$ permutations
out of all possible $n^k$ mappings of the $k$ random variables.
Thus, to get the probability of at least one $\pm d$ birthday collision, take the compliment of the probability of having no $\pm d$ birthday-collisions.
The next result follows.

\begin{theorem}[\cite{N}]
{\sf
\begin{eqnarray*}
B_d(n,k) & = & 1 - {n - (k-1)(d -1) \choose k} \ \frac{k!}{n^k} \ \ = \ \ 1 - \frac{(n - (k-1)(d -1))!}{(n - (k-1)(d -1) -k)!} \cdot \frac{1}{n^k},
\end{eqnarray*}
for $n \geq (k-1)(d-1) +k$ and $B_d(n,k) = 1$ otherwise.
}
\label{Naus}
\end{theorem}

Using the bound $1 - x \leq e^{-x}$ on Naus' result gives $k$ of about $0.83\sqrt{\frac{n}{d-4}}$, see~\cite{N}.
Also~\cite{DM} approximate $k$ to about $1.2 \sqrt{\frac{n}{2d+1}}$ for the cyclic version.

Note, Theorem~\ref{Naus} with $d = 1$ gives the solution to the standard birthday problem of Definition~\ref{standardBirthday}.
That is, a span of $d = 1$ and blocks of size $d - 1 = 0$.

The falling factorial is
\begin{eqnarray*}
m^{\underline{k}}	& = & m(m-1) \cdots (m-k+1) \ \ = \ \ {m \choose k} k!
\end{eqnarray*}

In these terms, Theorem~\ref{Naus} may be expressed as $B_d(n,k) = 1 - \frac{(n - (k-1)(d -1))^{\underline{k}}}{n^k}$.

\noindent
The next classic result is important.

\begin{lemma}[Classical]
{\sf
Let $m \geq k \geq 1$.
The falling factorial $m^{\underline{k}}$ is the number of injective mappings of $k \geq 1$ elements to the range $[m]$.
}
\label{surjective}
\end{lemma}

The next definition is based on~\cite{CN,NS,CDM}.

\begin{definition}[Boy-Girl Birthdays] % P...?
{\sf
Consider $n$ days in a year and two sets of distinctly labeled iid uniform random variables all with range $[n]$: $g$ of these variables are girls and $b$ 
of these variables are boys.
Then $B(n,b,g)$ is the probability at least one girl and one boy have a birthday-collision.
}
\end{definition}

For instance, in $n$ days, $\displaystyle \argmin_{\substack{k = b+g\\b=g}} \left\{ B(n,b,g) \geq \frac{1}{2} \right\}$ gives the value $k = b+g$ 
and $b=g$ so there is a probability  of $\frac{1}{2}$ 
where at least one girl and one boy have the same birthday.

Stirling numbers of the second kind~\cite{GKP} count the number of non-empty partitions of a given set.
For example given the set $[m]$, the number of partitions of $[m]$ into $i$ non-empty subsets is 
${m \brace i}$.

Due to their nature, it is common to define Stirling numbers of the second kind recursively~\cite{GKP}:
${m \brace i} = i {m-1 \brace i} + {m-1 \brace i-1}$
with the base cases
${m \brace 1} = 1$ and ${m \brace m} = 1$.
Finally, ${m \brace m+i} = 0$ for any $i > 0$.
As an example,

\begin{eqnarray*}
\left\{ \begin{array}{c}
	       	3\\
		2
	\end{array} \right\}
& = &  \left| \Bigl\{ \bigl\{ \{ 1,2 \}, \{ 3 \} \bigr\}, \bigl\{ \{ 1,3 \}, \{ 2 \} \bigr\}, \bigl\{ \{ 1 \}, \{ 2, 3 \} \bigr\} \Bigr\} \right| \ \ = \ \ 3.
\end{eqnarray*}

The next classical equality counts the number of functions from $[n]$ elements to $[m]$ elements, $m \geq n$,
\begin{eqnarray}
m^n & = & \displaystyle \sum_{i=1}^{n} \left\{ \begin{array}{c}
						n\\
						i
						\end{array}
				\right\} m^{\underline{i}} \label{EQ1}
\end{eqnarray}

\noindent
expressed as the number of non-empty $i$ partitions of the $[n]$ elements and the number of surjections from the $i$ partitions by Lemma~\ref{surjective}.

\begin{theorem}[\cite{CN,NS}]
{\sf
Consider $n$ days in a year and two sets of distinctly labeled iid uniform random variables all with range $[n]$: $g$ random variables are girls and $b$ random 
variables are boys.
Then $B(n,b,g)$ is the probability at least one girl and at least one boy have a birthday-collision and
\begin{eqnarray*}
B( n, b, g ) & = & 1 - \frac{1}{n^{b+g}} \sum_{i=1}^{g}		(n-i)^{b}
								\left\{ \begin{array}{c}
	       			 	    			g\\
								i
								\end{array} \right\}
								n^{\underline{i}}.
\end{eqnarray*}
}
\label{BG}
\end{theorem}

The next Lemma is from~\cite{CN,W}.

\begin{lemma}[\cite{CN,W}]
{\sf
Consider $n$ days in a year and two sets of distinctly labeled iid random variables all with range $[n]$: $g$ random variables are girls and $b$ random variables are boys.
Then $B(n,b,g)$ is the probability
that at least one girl and at least one boy have a birthday-collision and
\begin{eqnarray*}
B( n, b, g ) & = & 1 - \frac{1}{n^{b+g}} \sum_{i=1}^{g} \sum_{j=1}^{b}		
								\left\{ \begin{array}{c}
	       			 	    			b\\
								j
								\end{array} \right\}
								\left\{ \begin{array}{c}
	       			 	    			g\\
								i
								\end{array} \right\}
								n^{\underline{i+j}}.
\end{eqnarray*}
}
\end{lemma}

\subsection{Wash sale Example 1: Same Day Purchase and Sale}
\label{Wash-sale-example-1}

Consider a portfolio $\Pi = \{ a_1, \cdots, a_k \}$ where $a_i: k \geq i \geq 1$ is asset (security) $i$ held in $\Pi$.
At the end of business on day $\ell$, consider portfolio $\Pi_{\ell} = \{ a_{1,\ell}, \cdots, a_{k,\ell} \}$ the market value of asset $i$ in $\Pi_{\ell}$ 
is $|a_{i,\ell}|$ and the total value of $\Pi_{\ell}$ is $|\Pi_{\ell}| = \sum_{i=1}^{k} |a_{i,\ell}|$.
Just before the start of each tax year, asset $i$ has market value $|a_{i,0}|$ and $\Pi$ has total market value $|\Pi_0|$.
Assume each asset is sufficiently liquid so our purchases or sales do not impact its market price.

Suppose portfolio $\Pi$ has $T$ total iid uniform and random transactions during the business days of one calendar year.
Assume trades are distributed on an asset-weighted basis from the initial weight of each asset 
in the portfolio just before the trading year commences.
Thus, just prior to the first trading day and with no other information, asset $a_{i}$ is expected to have $t(i) = T \frac{|a_{i,0}|}{|\Pi_0|}$ trades 
in one year.

Take $t(i)$ transactions and define the independent Rademacher\footnote{We used Bernoulli random variables for $\{ 0,1 \}$ outcomes and we use Rademacher for $\{-1, +1 \}$ 
outcomes.} 
random variables $\eta_1, \cdots, \eta_{t(i)}$ representing buys or sells of portions of 
asset class $i$ in portfolio $\Pi$:
\begin{eqnarray*}
\eta_j & = & \left\{ \begin{array}{c}
    			+1 \ \ \mbox{ if transaction $j$ is a buy of asset } i \\
			-1 \ \ \mbox{ if transaction $j$ is a sell of asset } i \\
		\end{array} \right.
\end{eqnarray*}

\noindent
for $j: t(i) \geq j \geq 1$.
That is, the $b$  independent Rademacher random variables where $\eta_j = +1$ represent buys (boys)  and the $g$ random variables where $\eta_j = -1$ represent sells (gals).

To apply a suitable version of Chernoff's bound~\cite[Appendix A]{AS} where
$\Prob[\eta_j = +1] = \Prob[\eta_j = -1] = \frac{1}{2}$, then for any $c > 0$

\begin{eqnarray*}
\Prob\left[ \left( \eta_1 + \cdots + \eta_{t(i)} \right) > c \right] & < & e^{-c^2/(2t(i))}.
\end{eqnarray*}
So, for example, take $c =1$, then $|b-g| \leq 1$ holds with high probability as $t(i)$ gets large.
Of course, as $t(i)$ gets large, the likelihood of wash sales increases.
That is, the total number of buys and sells is expected to converge to be about the same as the total number of transactions grows.
However, along the way, the number of buys or sells may not be as balanced~\cite{F,NV}.

Select the probabilities that the number of buys and sales are the same, given $t(i)$ total trades, in asset class $a_i$ are:

\begin{center}
\begin{tabular}{ l || c c c c c }
   $t(i)$ 		& 10 & 20 & 30 & 40 & 50 \\ \hline
   $e^{-1/(2t(i))}$	& 0.951 & 0.975 & 0.983 &	0.987 &	0.990 
 \end{tabular}
\end{center}

Let $h$ be half the total trades $t(i)$ in day $i$.
That is, $h \leftarrow t(i)/2$.
Assuming $n \in \{ 252, 365 \}$ trading days gives the probabilities of same-day girl-boy birthday collisions for a single asset-type as:

\begin{center}
\begin{tabular}{ l || c c c c c c c c }
   $h$			& 1		& 5		& 10		& 15		& 20		& 25		& 30	  & 35  \\ \hline
   $B( 252, h, h )$	& 0.0040	& 0.0946	& 0.3280	& 0.5909	& 0.7957	& 0.9162	& 0.9717  & 0.9921 \\
   $B( 365, h, h )$	& 0.0027	& 0.0663	& 0.2399	& 0.4605	& 0.6660	& 0.8196	& 0.9150 & 0.9650
 \end{tabular}
\end{center}

In fact, $B( 252, 13, 13 ) = 0.4891$ and $B( 252, 14, 14 ) = 0.5410$.
So, considering only equal numbers of sales and buys over $n=252$ days of the same asset type, 14 girls and 14 boys is the first case where
there is greater than a 50\% chance of a (same-day) boy-girl birthday collision.

Assuming the portfolio $\Pi$ already holds this single asset type, a boy-girl collision only is a necessary
condition for a wash sale.
A birthday collision must be accompanied by a sale at a loss and a repurchase of substantially the same security within~30 calendar days.

%Given independent, uniform, and random trading of the portfolio $\Pi$ on an asset-weighted basis, then
%the probabilty of a same-day wash sales is 
%
%
%365 days and 35 boys and girls: 0.9650210893
%

%
%
%

\section{General Wash Sales}
\label{General-Wash-Sales}

Necessary conditions are given here for wash sales where a purchase and sale are within $\pm d$ calendar days.
Since the purchase and sale are not known to be at a loss while keeping substantially the same portfolio
before and after the $\pm d$ birthday collision.

\begin{definition}[Boy-Girl $\pm d$ Birthdays]
{\sf
Consider $n$ days in a year, spans of $\pm d$ days, and two sets of distinctly labeled iid uniform random variables all with range $[n]$: $g$ random variables are girls and 
$b$ random variables are boys. Then $B_d(n,g,b)$ is the probability
at least one girl and one boy are mapped to less than $d$ days of each other.
}
\end{definition}

For example, starting with $n,d$ and $k=g+b$ and $g = b$, then $\displaystyle \argmin_k \left\{ B_d(n,\textstyle\frac{k}{2}, \textstyle\frac{k}{2}) \geq \frac{1}{2} \right\}$ 
gives $k$ so there is a probability  of $\geq \frac{1}{2}$ so at least one girl and one boy have $\pm d$-birthday collisions.

The next result is based on~\cite{N}, \cite{NS}, \cite{CN}, and~\cite{W}.

\begin{theorem}
{\sf
Consider $n$ days in a year, a span of $\pm d$ days, and two sets of distinctly labeled iid uniform random variables all with range $[n]$: $g$ random variables are girls
and $b$ random variables are boys. Then $B_d(n,g,b)$ is the probability
at least one girl and one boy have a $\pm d$ birthday-collision and:
\begin{eqnarray*}
B_d(n,g,b) & = & 1 - \frac{1}{n^{b+g}} \sum_{i=1}^{b} \sum_{j=1}^{g}  \left\{ \begin{array}{c}
	       			 	    			b\\
								i
								\end{array} \right\}
								\left\{ \begin{array}{c}
	       			 	    			g\\
								j
								\end{array} \right\}
								(n - (i+j-1)(d -1))^{\underline{i+j}}.
\end{eqnarray*}
}
\label{BG_Lower}
\end{theorem}
\begin{proof}
This proof calculates the probability of not having no boy-girl $\pm d$ birthday collisions.
That is, one minus the probability of no boy-girl $\pm d$ birthday collisions.
This gives the probability of at least one boy-girl $\pm d$ birthday collision.

Given $n$ days, a $\pm d$ span, and iid uniform random variables separated into $g$ (girls) random variables and 
$b$ (boys) random variables. Then the total number unconstrained mappings
of the $b$ and $g$ variables to $[n]$ is $n^{b+g}$ giving the denominator in front of the double sum.

The value $B_d(n,g,b)$ is not impacted if either any number of boys have the same birthday or separately any number of girls have the same birthday.
Rather $B_d(n,g,b)$ is impacted by boy-girl collisions.
Therefore, consider partitions of $b$ boys and $g$ girls.
To prevent the girls' partitions and boys' partitions from colliding into $\pm d$ spans of the same range,
count the number of places these $i$ and $j$ non-empty partitions may be mapped so there is no $\pm d > 1$ birthday-collision.
By Lemma~\ref{surjective} there are 
\begin{eqnarray}
(n - (i+j-1)(d -1))^{\underline{i+j}}
& = & {n - (i+j-1)(d -1) \choose i+j} (i+j)!
\label{EQ7}
\end{eqnarray}
\noindent
injective functions to $[n]$ for sets of $i \in [b]$ boys and sets of $j \in [g]$ girls with $(i+j-1)$ blocks of $(d-1)$ contiguous days with no
boy or girl in them.

Now, consider placing the $i$ and $j$ partitions in separate locations among the $(n - (i+j-1)(d -1))^{\underline{i+j}}$ function mappings to $[n]$.
The $i$ partitions of $[b]$ where each partition is in a different location and $j$ partitions of $[g]$ where each partition
is also in a different location by Equation~\ref{EQ7}.
That is, given $i \in [b]$ and $j \in [g]$, then the product ${b \brace i}{g \brace j}$
is the total number of injective mappings of boys to $i$ non-empty partitions and independently the number of injective mappings of girls to $j$ non-empty partitions.

This completes the proof.
\end{proof}

\subsection{Wash sale Example 2: $d = \pm 30$ Calendar Days}
\label{Wash-sale-example-2}

Start with the same setup as the previous wash sale example from subsection~\ref{Wash-sale-example-1}.

Let $h$ be half the total trades $t(i)$ in day $i$.
That is, $h \leftarrow t(i)/2$.
Assuming $n \in \{ 252, 365 \}$ trading days and $d = \pm 30$ calendar days gives the probabilities of 
girl-boy $\pm 30$-day birthday-collisions for a single asset type is:

\begin{center}
\begin{tabular}{ l || c c c c c c c c }
   $h$		      		& 1		& 2		& 3		& 4		  \\ \hline
   $B_{30}( 252, h, h )$	& 0.220		& 0.819		& 0.994		& 0.99998	  \\
   $B_{30}( 365, h, h )$	& 0.155		& 0.667		& 0.953		& 0.99840
 \end{tabular}
\end{center}

Consider only a single asset type.
The intuition behind these probabilities is straight-forward.
For instance, consider $n = 365$ days and to avoid boy-girl collisions each girl and boy must be separated by at least 30 days before and after their 
birthday from the other gender.
So the $365$ days may be broken into about six blocks of about 60 days.

\section{Wash Sale and Integral Capital Gains and Losses}
\label{VarCGCL}

Capital gains or capital losses may be rounded to the nearest integer for US tax calculations. Provided all trades are rounded.
Rounding drops the cents portion for gains whose cents portion is 50-cents or below.
Rounding adds a dollar to the dollar portion of gains whose cents portion is greater than 50 cents while dropping the cents portion.
Losses work the same way.
Gains and losses must all be rounded or none must be rounded.
So, from here on, let all gains or losses be integers.

Long term capital gains and losses are aggregated and at the same time short term capital gains and losses are aggregated. 
At the end of the tax year the long term and short term aggregates are added together to get the final capital gain or loss for taxation.

The focus here is capital gains or losses for capital assets that may have wash sales.
Wash sales are losses, but losses may offset gains.
The study of options and their associated premiums is classical~\cite{H} and we do not address it here.
So, option premiums are ignored.

In a portfolio, individual capital gain values and individual capital loss values are usually distinct.
Though rare, identical capital gains and capital losses are possible.
Identical capital gains or losses are possible for portfolios built using options.
We are ignoring option premiums.
That is, asset purchases may be done via the exercise of cash-covered American-style put options. Also asset sales may be done via
the exercise of American-style covered-call options.
In these cases with options that become in-the-money, a portfolio manager has no control of the asset sales or purchases or timing
of such trades.
See Figure~\ref{GainWashSale}.

Most often, put or call option strike prices are at discrete increments. For example, many put and call equity options have strike prices 
in \$5 or \$10 increments.
Suppose a portfolio is built only using the exercise of American-style options.
Many asset gains and losses may be for identical amounts.
Of course, this depends on the size of the underlying positions or the number of options written.
Options with the same expiry on identically sized underlying assets may have very different values~\cite{H}.

\begin{figure}[h]
\resizebox{\textwidth}{!}{%
\begin{tabular}{ l | l | l | l | l }
date $d_1$			& date $d_2$		&	date $d_3$		&	date $d_4$	&	date $d_5$	\\ \hline
Sell cash-covered 		& Put $p_1$ exercised	& Covered-call $c_1$ 		& Call $c_1$ exercised	& No wash sale		\\
put $p_1$ at strike 		& to `purchase' the 	& sold for strike  		& triggering a sale	&			\\
price \$100 on $d_2$		& asset for \$100	& price \$110 			& giving a capital 	&			\\
				&			& and expiry $d_4$		& gain of \$10		&			\\ \hline
Sell cash-covered 		& Put $p_2$ exercised	& Covered-call $c_2$		& Call $c_2$ exercised	& Sell cash-covered put	\\
put $p_2$ at strike 		& to `purchase' the	& sold for strike  		& triggering a sale	&  with strike price  	\\
price \$100 on $d_2$		& asset for \$100	& price \$90 and		& giving a capital 	& \$90 and if exercised	\\			
				&			& expiry on $d_4$		& loss of \$10		& within 30 days, then	\\
				&			&				&			& it is a wash sale
\end{tabular}}
\caption{A Potential Wash Sale with American-style Options}
\label{GainWashSale}
\end{figure}

In such option-based portfolios assume uniform, independent, and random capital gains and capital losses.
This may be modeled by the Littlewood-Offord Problem.

Definition~\ref{LO} is classical and extensive discussion may be found in the likes of~\cite{TV2006,TV}.
It is based directly on~\cite{LO,TV2006,TV}

\begin{definition}[Littlewood-Offord Problem]
{\sf
The integer Littlewood and Offord's problem is given an integer multi-set $V = \{ v_1, v_2, \cdots, v_n \}$ where 
$v_i \geq 1, \ \forall i \in [n]$ and $S_v = \xi_1 v_1 + \xi_2 v_2 + \cdots + \xi_n v_n$ so each $\xi_i$ is such that 
$\Prob[\xi_i = -1] = \Prob[\xi_i = +1] = \frac{1}{2}$, for $i \in [n]$, then what is $\max_{x \in \IZ} \Prob \left[ S_v = x \right]$?
}
\label{LO}
\end{definition}

Assuming equal probability of gains and losses and no drift~\cite{H}.
Given an integer multi-set $V = \{ v_1, v_2, \cdots, v_n \}$ so $v_i \geq 1, \forall i \in [n]$. The multi-set $V$ represents capital gains 
and capital losses.
Capital gains and capital losses are all from sales.
The iid Rademacher random variables $\xi_i \in \{ +1, -1 \}$ determine if a $v_i$ is a capital gain or loss.
All $v_i$ are positive since all the Rademacher variables have range $\{ -1, +1 \}$, see also \cite{E1945} and \cite{NV}.

Over a tax year, the total capital gain or loss is
\begin{eqnarray*}
S_v & = & \xi_1 v_1  + \xi_2 v_2  + \cdots + \xi_n v_n.
\end{eqnarray*}

In an optimal solution of this version of the Littlewood-Offord problem, \cite{E1945} showed
the $n$-element multi-set $V = \{ 1, 1, \cdots, 1 \}$ has $\max_{x \in \IZ} \left\{ \Prob \left[ S_v = x \right] \right\} = O\left( \frac{1}{\sqrt{n}} \right)$.

The next lemma's proof follows immediately from the linearity of expectation given Rademacher random variables.
See, for example,~\cite{AS}.

\begin{lemma}
{\sf
Consider any integer multi-set $V = \{ v_1, v_2, \cdots, v_n \}$ where $v_i \geq 1, \forall i \in [n]$ and the random variable 
$S_v = \xi_1 v_1 + \xi_2 v_2 + \cdots + \xi_n v_n$, 
where $\Prob[\xi_i = -1] = \Prob[\xi_i = +1] = \frac{1}{2}$, for all $i \in [n]$, then $\E[S_v] = 0$.
}
\label{Mean0}
\end{lemma}

For any Rademacher random variable $\xi_i$, it must be $\E[\xi_i] = 0$ and $\E[\xi_i^2] = 1$.
Since $v_i$ is constant $\sigma^2_{\xi_i v_i} = \E[\xi_i^2 v_i^2] - \E[\xi_i v_i]^2 = v_i^2$.
Thus, a proof of the next theorem follows since the variance of a sum of independent random variables is the sum of the variances.

\begin{theorem}
{\sf
Consider any non-negative integer vector $v$ and the random variable $S_v = \xi_1 v_1 + \xi_2 v_2 + \cdots + \xi_n v_n$,
where $\Prob[\xi_i = -1] = \Prob[\xi_i = +1] = \frac{1}{2}$, for all $i \in [n]$, then $\E[S_v^2] = v_1^2 + v_2^2 + \cdots + v_n^2$
and $\sigma_{S_v} = \sqrt{v_1^2 + v_2^2 + \cdots + v_n^2}$.
}
\label{SumSquareTheorem}
\end{theorem}

Thus, the lowest variance, $\sigma_{S_v}^2$, for the integer Littlewood-Offord problem occurs exactly when $V = \{ 1, 1, \cdots, 1 \}$ and
$|V| = n$.
Assuming the $\xi_i, \ \forall i \in [n]$ are all Rademacher random variables, then
$\Prob_{x \in \IZ} [S_v = x]$ is maximized~\cite{TV2006,TV,NV} as $O(1/\sqrt{n})$ and $\sigma_{S_v} = \sqrt{n}$.

Theorem~\ref{SumSquareTheorem} implies the next corollary.

\begin{corollary}
{\sf
Assume $1 = v_1 = v_2 = \cdots = v_n$ and $S_v = \xi_1 v_1 + \xi_2 v_2 + \cdots + \xi_n v_n$
where $\Prob[\xi_i = -1] = \Prob[\xi_i = +1] = \frac{1}{2}$, for all $i \in [n]$, 
then the standard deviation of $S_v$ is $\sigma_{S_v} = \sqrt{n}$.
}
\label{linearCorollary}
\end{corollary}

Corollary~\ref{linearCorollary} highlights an exceptional case where all capital gains and capital losses are the same.
Wash sales require the loss and gain to be from essentially the same security.

The generality of Theorem~\ref{SumSquareTheorem} asserts large variances too.
Consider the set $V = \{ 2^0, 2^1, \cdots, 2^{n-1} \}$,
then by Theorem~\ref{SumSquareTheorem}, $\sigma_{S_v}^2 = \sum_{i=0}^{n-1} 2^{2i} = \frac{2^{2n} -1}{3}$.
This last equality follows since the sum is a geometric series.

\begin{definition}[Distinct sums of a set or multi-set $V$]
{\sf
Consider a set or multi-set $V = \{ v_1, v_2, \cdots, v_n \}$ and
let each element of the lists $H_1 = \langle \hat{\xi}_{1,1}, \hat{\xi}_{2,1}, \cdots, \hat{\xi}_{n,1} \rangle$ and 
$H_2 = \langle \hat{\xi}_{1,2}, \hat{\xi}_{2,2}, \cdots, \hat{\xi}_{n,2} \rangle$ be fixed values from $\{ -1, +1 \}$.
The two sums of $V$,

\begin{eqnarray*}
s_{v,1} & = & \hat{\xi}_{1,1} \, v_1 + \hat{\xi}_{2,1} \, v_2 + \cdots + \hat{\xi}_{n,1} \, v_n\\
s_{v,2} & = & \hat{\xi}_{1,2} \, v_1 + \hat{\xi}_{2,2} \, v_2 + \cdots + \hat{\xi}_{n,2} \, v_n,
\end{eqnarray*}

\noindent
are {\em distinct} iff there is some $\hat{\xi}_{i,1} \neq \hat{\xi}_{i,2}$, for $i \in [n]$.
}
\end{definition}
%
%If two distinct sums add to the same value so $s_{v,1} = s_{v,2}$, then these distinct sums are the same.
%

Given any multi-set of positive integers $V = \{ v_1, v_2, \cdots, v_n \}$, enumerate all $2^n$ distinct sums
as $s_v[1] \geq s_v[2] \geq \cdots \geq s_v[2^n]$, for example, see Figure~\ref{Sv-binomial}.
Given any set of positive integers $V = \{ v_1, v_2, \cdots, v_n \}$, where none of the $2^n$ distinct sums add to the same value gives
$s_v[1] > s_v[2] > \cdots > s_v[2^n]$. 

An important observation by \cite{E1945}, is that for any fixed sum $s$ the values $s + v_{i}$ and 
$s - v_{i}$ differ by $2 v_{i}$.
Next, this observation is used to 
show the set $V = \{ 2^0, 2^1, \cdots, 2^{n-1} \}$ has no distinct sums that add to the same value.

In particular, take any distinct sums $s_{v,1}$ and $s_{v,2}$ with associated fixed values $\hat{\xi}_{i,1} \in \{ -1, +1 \}$ 
and $\hat{\xi}_{i,2} \in \{ -1, +1 \}$, respectively, for all $i \in [n]$.
Suppose, for the sake of a contradiction, that $s_{v,1} = s_{v,2}$.
Building on Erd\H{o}s' observation, the values $s_{v,1}$ and $s_{v,2}$ may be written as
$s_{v,1} = 2^{n} - 1 - 2 \, m_1$ where $m_1 = \sum_{i \in I_1} 2^{i-1}$ and $I_1 = \{ \, i : \hat{\xi}_{i,1} = -1 \, \}$
and likewise $s_{v,2} = 2^{n} - 1 - 2 \, m_2$ where $m_2 = \sum_{i \in I_2} 2^{i-1}$ and $I_2  = \{ \, i : \hat{\xi}_{i,2} = -1 \, \}$, for all $i \in [n]$.
Finally, the uniqueness of binary-number representations means $m_1 = m_2$ which in turn means $\hat{\xi}_{i,1} = \hat{\xi}_{i,2}$, for all $i \in [n]$.
So, in fact, the sums $s_{v,1}$ and $s_{v,2}$ are equal, giving a contradiction.

Thus, the set $V = \{ 2^0, 2^1, \cdots, 2^{n-1} \}$ satisfies the antecedent of the next theorem. 

\begin{theorem}
{\sf
Among all sets of distinct positive integers where no two distinct sums add to the same value, 
the set $V = \{ 2^0, 2^1, \cdots, 2^{n-1} \}$ has a minimal sum $s_v = v_1 + v_2 + \cdots + v_n = 2^{n} -1$.
}
\label{Distinct-Theorem}
\end{theorem}

\begin{proof}
Suppose, for the sake of a contradiction, that $s_v = v_1 + v_2 + \cdots + v_n < 2^n -1$ for some set of distinct positive integers $V = \{ v_1, v_2, \cdots, v_n \}$
where no two distinct sums add to the same value.

Take the next enumeration of the $2^n$ distinct sums, $s_v[1] > s_v[2] > \cdots > s_v[2^n]$, and
by our supposition, $2^n -2 \geq s_v[1]$ and $s_v[2^n] \geq -2^n +2$, so $s_v[1] - s_v[2^n] \leq 2^{n+1} - 4$.

Let $\{ \, s_{v,1}, s_{v,2} \} \subseteq \{ \, s_v[1], s_v[2], \cdots, s_v[2^n]  \, \}$
where sum $s_{v,1}$ has the list of fixed values $H = \langle \hat{\xi}_{1,1}, \hat{\xi}_{2,1}, \cdots, \hat{\xi}_{n,1} \rangle$
so that $s_{v,1} = \langle v_1, v_2, \cdots, v_n \rangle \cdot H$, where $\cdot$ is the vector dot product.
Likewise, the sum $s_{v,2}$ has the list of fixed values $\langle \hat{\xi}_{1,2}, \hat{\xi}_{2,2}, \cdots, \hat{\xi}_{n,2} \rangle$.

The difference of any two distinct sums $s_{v,1} -  s_{v,2}$ 
must be even since any fixed values $\hat{\xi}_{i,1} \in \{ -1, +1 \}$ 
and $\hat{\xi}_{i,2} \in \{ -1, +1 \}$, for $i \in [n]$, are so that,
\begin{eqnarray*}
\hat{\xi}_{i,1} - \hat{\xi}_{i,2} \in \{ \ 0, -2, +2 \ \},
\end{eqnarray*}
giving
\begin{eqnarray*}
s_{v,1} -  s_{v,2} & = & \sum_{i=1}^{n} v_i \left( \hat{\xi}_{i,1} - \hat{\xi}_{i,2} \right)
\end{eqnarray*}
which must be even.

Starting from $s_v[1]$ and going to $s_v[2^n]$ contains $2^n-1$ intervals.
Since all $s_v[i]$, for $i \in [2^n]$, are different and their differences must be even
so $s_v[1] - s_v[2^n]$ spans at least $2 (2^{n} -1) = 2^{n+1} -2$. 
That is, $s_v[1] - s_v[2^n] \geq 2^{n+1} -2$.
This gives a contradiction of the assumption $s_v[1] - s_v[2^n] \leq 2^{n+1} - 4$, completing the proof.
\end{proof}

Given a set of distinct positive integers $V$ where $|V| = n$, Theorem~\ref{Distinct-Theorem} indicates that 
$\max_{x \in \IZ} \left\{ \Prob [S_v = x] \right\} \leq \frac{1}{2^{n}}$.
So in the case where all distinct sums of $V$ add to different values, erasing a wash sale loss may have a very large impact.
In particular, the multi-set $V = \{ 1, 1, \cdots, 1 \}$ has largest loss $s_v[2^n] = -n$, where 
Theorem~\ref{Distinct-Theorem} indicates $V = \{ 2^0, 2^1, \cdots, 2^{n-1} \}$ has the largest loss $s_v[2^n] = -2^n +1$.
In this case, when no distinct sums add to the same value, let $U = \{ 2^n - (2i -1) : i \in [2^{n}] \}$
giving $\max_{x \in U} \left\{ \Prob [S_v = x] \right\} = \frac{1}{2^{n}}$. 
Assuming wash sales occur with the same random and uniform probability among all losses, the expected disallowed loss
is $\frac{2^n -1}{n}$.
This is because all losses are of the form $-(2^{i-1})$, for $i \in [n+1]$, and by assumption these losses all have the same 
probability of occurring.

%
%Given $V = \{ v_1, v_2, \cdots, v_{2n} \}$ where $v_{n+i} > v_1 + v_2 + \cdots + v_{n}$ for all $i \in [n]$, then
%

Since $\E[S_v] = 0$ by Lemma~\ref{Mean0}, Littlewood-Offord results are useful for understanding 
likely values for $S_v$.
That is, $\max_{x \in \IZ - \{ 0 \}} \left\{ \Prob [S_v = x] \right\}$ gives most likely capital gains or losses
outside of the expected value $\E[S_v] = 0$.
None of the $s_v$ values in Figure~\ref{Sv-binomial} are~$0$, but if $V$ has an even number of $1$s, then the most common value is~$0$.

The following tail bound is given by \cite{SJMS} where $\| v_1, v_2, \cdots, v_n \|_2 = \sqrt{v_1^2 + v_2^2 + \cdots + v_n^2}$,

\begin{eqnarray*}
\Prob \left[ \sum_{i=1}^{n} \xi_i v_i > t \| v_1, v_2, \cdots, v_n \|_2 \right]	& \leq & e^{-t^2/2}\\
\Prob \left[ S_v > t \sqrt{v_1^2 + v_2^2 + \cdots + v_n^2} \right] 		& \leq & e^{-t^2/2}\\
\Prob \left[ S_v > t \sigma_{S_v} \right]  	    	   			& \leq & e^{-t^2/2}
\end{eqnarray*}

\noindent
Since by Theorem~\ref{SumSquareTheorem}, $\sigma_{S_v} = \sqrt{v_1^2 + v_2^2 + \cdots + v_n^2}$.

Suppose $V = \{ 1, 1, \cdots, 1 \}$ and $|V|$ is odd.
Since no sum of $V$ is $0$, there are $\frac{n}{2}$ capital gains and $\frac{n}{2}$ capital losses.
This means if $S_v = t \sigma$, then there are $\frac{n}{2} + \frac{t \sigma}{2}$ capital gains and $\frac{n}{2} - \frac{t \sigma}{2}$ 
capital losses.
Losses are necessary for wash sales.
Therefore, the bound $\Prob \left[ S_v > t \sigma_{S_v} \right] \leq e^{-t^2/2}$ gives the probability there are at least $\frac{t \sigma_{S_v}}{2}$ more gains 
than losses. That is, there are $\frac{t \sigma_{S_v}}{2}$ fewer opportunities for wash sales.

\begin{figure}
\begin{center}
\begin{tabular}{|l|l|} \hline
$V =\{ \, 1,1,1 \, \}$		& $s_v$    	\\ \hline
$+1+1+1$			& $s_v[1] = 3$  \\
$-1+1+1$			& $s_v[2] = 1$  \\
$+1-1+1$			& $s_v[3] = 1$  \\
$+1+1-1$			& $s_v[4] = 1$  \\
$-1-1+1$			& $s_v[5] = -1$ \\
$-1+1-1$			& $s_v[6] = -1$ \\
$+1-1-1$			& $s_v[7] = -1$ \\
$-1-1-1$			& $s_v[8] = -3$ \\
\hline
\end{tabular}
\end{center}
\caption{\bf The case where $v_1 = v_2 = v_3 = 1$ and $s_v$ is made of ${3 \choose 0}, {3 \choose 1}, {3 \choose 2}, {3 \choose 3}$ elements of $3,1,-1,-3$, respectively}
\label{Sv-binomial}
\end{figure}

Following Figure~\ref{Sv-binomial}, given $|V| = n$ then $s_v[1] = n$ is the case with zero capital losses. Likewise, $s_v[2^n] = -n$ is the case with
zero capital gains.
By Lemma~\ref{Mean0}, since $\E[S_v] = 0$ and $s_v[1] + \cdots + s_v[2^n] = 0$, thus $-n = s_v[2] + \cdots + s_v[2^n]$.
Also suppose a single wash sale disallows a capital loss among all identical capital gains and losses.
The single wash sale disallows a single capital loss giving the expected capital gain or loss:
\begin{eqnarray*}
\frac{(s_v[2] + 1) + (s_v[3] + 1) + \cdots + (s_v[2^n] +1)}{2^{n}-1}.
\end{eqnarray*}
The term $s_v[1]$ is excluded since it has no losses, hence no wash sales.

The  boy-girl $\pm 30$ birthday problem gives a necessary condition for wash sales of substantially identical securities.
Recall $B_{30}(252,g,b)$ is the probability of at least one boy-girl $\pm 30$ birthday collision, so $1- B_{30}(252,g,b)$ is the probability of
no such birthday collision.

Given any number of boy-girl $\pm 30$ birthday collisions of the same security and suppose these birthday collisions produce at most a single wash sale.
In this case let $G$ be a total taxable gain or loss where all gains and losses are the same. Suppose these gains and losses are all~$1$.
This gives,

\begin{eqnarray*}
\E[ \, G \, ] & = & \left(1 - B_{30}(252,g,b)\right) \frac{s_v[1] + s_v[2] + \cdots + s_v[2^n]}{2^n} \\
		&  & \	+ \ B_{30}(252,g,b) \frac{(s_v[2] + 1) + \cdots + (s_v[2^n] +1)}{2^n-1} \\
		& = & \frac{1 - B_{30}(252,g,b)}{2^n-1}(0) + \frac{B_{30}(252,g,b)}{2^n-1}\left( 2^n-1 -n \right)\\
		& = & B_{30}(252,g,b)\left( 1 - \frac{n}{2^n -1} \right).
\end{eqnarray*}

\section{Conclusions and further directions}

Wash sales may be modeled in a number of ways. These include variations of the birthday problem and the capital gains of portfolios and
wash sales impact may be modeled using the Littlewood-Offord problem.

The $k$-armed bandit, see for example~\cite{Robbins} or~\cite{KV}, etc., appears 
to apply to wash sales and the birthday problem.
Robbins' discussion of maximizing expected value of sums of random variables selected from different distributions is applicable
to constructing portfolios by writing options.

\section{Acknowledgement}

Thanks to Noga Alon for insightful comments.

\pagebreak

%
%
%
%
%
%

%**************************************

%
%
%
%

%\bibliographystylei{eeetr} % unsrt
%
% http://tex.stackexchange.com/questions/82907/citet-and-citep-behaves-like-cite
%
%\begin {thebibliography}{Phillip G. Bradford 2064}
%{\color{myaqua}}

\end{document}